\documentclass[12pt]{article}

\usepackage{amssymb}
\usepackage{amsfonts}
\usepackage{amsmath}
\usepackage{amsthm}
\usepackage[english]{babel}
\usepackage{latexsym}
\usepackage{color}
\usepackage{url}
\usepackage{enumerate}
\usepackage{enumitem}
\usepackage{nicefrac}
\usepackage{mathrsfs}
\usepackage{comment}
\usepackage[hmargin=2cm,vmargin=1.9cm]{geometry}
\usepackage{bm}
\usepackage{bbm}

\usepackage{todonotes}
\usepackage{centernot}

\theoremstyle{plain}
\newtheorem{theorem}{Theorem}
\newtheorem{lemma}[theorem]{Lemma}

\newtheorem{proposition}[theorem]{Proposition}

\theoremstyle{definition}

\newtheorem{remark}[theorem]{Remark}

\theoremstyle{remark}

\numberwithin{equation}{section}

\newcommand{\E}{E}

\newcommand{\probp}{P}
\newcommand{\probq}{Q}
\newcommand{\R}{\mathbb{R}}
\newcommand{\N}{\mathbb{N}}
\newcommand{\Q}{Q}

\newcommand{\cF}{{\mathcal{F}}}

\newcommand{\cA}{\mathcal{A}}

\newcommand{\cP}{\mathcal{P}}

\newcommand{\ind}{\mathbbm 1}

\newcommand{\VaR}{\mathop {\rm VaR}\nolimits}

\newcommand{\ES}{\mathop {\rm ES}\nolimits}
\newcommand{\WCE}{\mathop {\rm WCE}\nolimits}

\makeatletter
\def\@fnsymbol#1{\ensuremath{\ifcase#1\or 1\or 2\or 3\or 4\or 5\or 6\or 7\or 8\else\@ctrerr\fi}}
\makeatother

\renewcommand{\epsilon}{\varepsilon}
\newcommand{\diff}{\mathrm{d}}
\newcommand{\dd}{\,\mathrm{d}}

\begin{document}

\title{An elementary proof of the\\ dual representation of Expected Shortfall\footnote{We thank Hans F\"{o}llmer, Alexander Schied and Ruodu Wang for helpful bibliographical suggestions and comments.}}

\author{
Martin Herdegen \\ \textit{University of Warwick}
\and
Cosimo Munari \\ \textit{University of Zurich}
}

\date{\today}

\maketitle

\begin{abstract}
We provide an elementary proof of the dual representation of Expected Shortfall on the space of integrable random variables over a general probability space. Unlike the results in the extant literature, our proof only exploits basic properties of quantile functions and can thus be easily implemented in any graduate course on risk measures. As a byproduct, we obtain a new proof of the subadditivity of Expected Shortfall.
\end{abstract}







\section{Introduction}

The debate on capital adequacy and solvency regulation in the past thirty years has been dominated by two competing risk measures: Value at Risk and Expected Shortfall. As is well known, Value at Risk became popular as part of the RiskMetrics package developed by J.P.\ Morgan in the 1990s with the aim to provide market participants with a set of techniques and data
to measure market risks in their portfolios; see \cite{RiskMetrics}. Shortly after, Value at Risk was chosen by the Basel Committee on Banking Supervision as the reference market risk measure in the Basel II framework and was later to become the reference credit risk measure in the Basel III framework as well as the reference metric used by European insurance regulators for the computation of solvency capital requirements within the Solvency II framework; see \cite{BaselII,BaselIIIb,SolvencyII}. The introduction of Value at Risk raised a number of concerns about its ability to properly capture (tail) risks and to create the right incentives towards portfolio diversification, which eventually led to the definition of Expected Shortfall in the early 2000s; see \cite{AcerbiTasche2002,ArtznerDelbaenEberHeath1999,BertsimasLaupreteSamarov2004,Pflug2000,RockafellarUryasev2000,RockafellarUryasev2002}. Currently, Expected Shortfall has replaced Value at Risk as the reference market risk measure in the Basel III framework and is employed to compute solvency capital requirements for insurance and reinsurance companies in the Swiss Solvency Test; see \cite{BaselIII,SST}. Both risk measures have been the subject of an intense research program that was aimed to uncover their relative merits and drawbacks both from a theoretical and empirical point of view; see, e.g., \cite{AcerbiSzekely2014,BarrieuScandolo2015,BignozziBurzoniMunari2020,BurzoniMunariWang2022,ContDeguestScandolo2010,KochMunari2016,KraetschmerSchiedZaehle2014,MaoWang2020,MunariWeberWilhelmy2022,WangZitikis2021,Weber2018,YamaiYoshiba2005,Ziegel2016}.

\smallskip

A key difference at a theoretical level is that Expected Shortfall is a coherent risk measure in the sense of \cite{ArtznerDelbaenEberHeath1999} whereas Value at Risk is not as it fails to satisfy the important property of subadditivity. Being coherent, Expected Shortfall can be equivalently described as a ``robust expectation'', i.e., as a supremum of expectations over a suitable family of probability measures. More precisely, let $(\Omega,\cF,\probp)$ be a probability space 
and denote by $L^1$ and $L^\infty$ the space of $\probp$-integrable and $\probp$-almost surely bounded random variables, respectively. 
Following \cite{FoellmerSchied2002}, we define the \emph{Expected Shortfall}, also called {\em Average Value at Risk}, of $X\in L^1$ at level $\alpha\in(0,1)$ by
\[
\ES_\alpha(X) := \frac{1}{\alpha}\int_0^\alpha \VaR_\beta(X) \dd \beta,
\]
where $\VaR_\beta(X) := -\inf\{x\in\R \,; \ \probp(X\le x)>\beta\}$ is the \emph{Value at Risk} of $X$ at level $\beta \in (0, 1)$, which coincides, up to a sign, with the upper quantile of $X$ at level $\beta$. We denote by $\cP$ the collection of all probability measures on $\cF$, and for $\alpha \in (0, 1)$ we set
\begin{equation}
\label{eq:cP alpha}
\cP_\alpha := \left\{\probq\in\cP \,; \ \probq\ll\probp, \ \frac{\dd\probq}{\dd\probp}\le\frac{1}{\alpha} \ \mbox{$\probp$-a.s.}\right\}.
\end{equation}
We can then express the Expected Shortfall of $X\in L^1$ as a ``robust expectation'' over $\cP_\alpha$, namely
\begin{equation}
\label{eq: ES intro}
\ES_\alpha(X) = \sup_{\probq\in\cP_\alpha}\E_\probq(-X).
\end{equation}
This representation corresponds to the classical ``Fenchel-Moreau-Rockafellar'' dual representation from convex analysis applied to Expected Shortfall; see \cite[Theorem 2.3.3]{Zalinescu} for a general formulation. The representation is informative {\em per se} and becomes a useful tool in a variety of applications featuring Expected Shortfall, e.g., to pricing, hedging, portfolio selection; see \cite{ArducaMunari2023,Cherny2008,HerdegenKhan2022,MadanCherny2010,PflugRuszczynski2005,Rockafellar1974,RockafellarUryasev2000,RockafellarUryasev2002,RuszczynskiShapiro2006}.

\smallskip

Historically, the intuition behind the dual representation of Expected Shortfall can be traced back to the link between Expected Shortfall and the Worst Conditional Expectation at level $\alpha\in(0,1)$, which, following \cite{ArtznerDelbaenEberHeath1999}, is defined for every $X\in L^1$ by
\[
\WCE_\alpha(X) := \sup_{A\in\cF,\,\probp[A]>\alpha}\E_\probp(-X\vert A).
\]
In a nonatomic setting, the Worst Conditional Expectation admits a dual representation in the form of the right hand side of \eqref{eq: ES intro}; this was established in  \cite[Example 4.2]{Delbaen2002}.\footnote{The result can already be found in the preprint version from March 2000.} This result then automatically delivers the dual representation of Expected Shortfall for atomless probability spaces because the two risk measures coincide in this setting; see, e.g., \cite[Theorem 6.10]{Delbaen2002}. However, as shown in \cite{AcerbiTasche2002}, the two risk measures (and their dual representations) do not coincide on general probability spaces, thereby requiring an independent study of Expected Shortfall in a general setting.

\smallskip

To the best of our knowledge, the first complete derivation of the dual representation of Expected Shortfall was obtained in \cite[Theorem 4.39]{FoellmerSchied2002} for bounded random variables. The proof is based on two steps. First, the equivalent formulation of Expected Shortfall as a tail conditional expectation, which was originally proved in \cite{AcerbiTasche2002,Pflug2000}, is used to show that Expected Shortfall dominates from above each expectation in \eqref{eq: ES intro}. Second, a Neyman-Pearson type argument is employed to show that Expected Shortfall coincides with one of those expectations for a suitable choice of the underlying probability so that one has actually equality in \eqref{eq: ES intro}. 
A similar two-step argument is used in \cite{PflugRoemisch2007} but the representation is stated for integrable random variables defined on a nonatomic probability space. In the aforementioned references the dual representation was also used to derive -- as a direct byproduct -- another important property of Expected Shortfall, subadditivity. This duality-based proof of subadditivity is included in the survey article \cite{EmbrechtsWang2015}, where it is said that ``[this proof] is probably the most mathematically advanced among all proofs in this paper''. In fact, the authors recommend it ``in an advanced course where the axiomatic theory of coherent risk measures is a point of interest''.

\smallskip

The goal of this short note is to provide an elementary proof of the dual representation \eqref{eq: ES intro} of Expected Shortfall, and, as a byproduct, a new proof of its subadditivity, for integrable random variables over a general probability space. Our approach only relies on basic properties of quantile functions and standard results from measure theory. In particular, it does not require the equivalent formulation of Expected Shortfall as a tail conditional expectation. We first obtain the desired representation for simple random variables. A straightforward limiting argument allows to extend it to bounded random variables. Finally, the continuity from above of quantile functions makes it possible to further extend it to all integrable random variables. The advantage of this multi-layer approach is that one can tailor to the reference audience the choice of the model space and, hence, the overall mathematical complexity of the argument (finite/general probability space, simple/bounded/integrable random variables). We believe that the proof in the present note is considerably simpler than the ones in the extant literature and can thus be successfully implemented in any graduate course on risk measures. In this pedagogical spirit, we collect all the basic properties of quantile functions that are used in the note in the appendix and provide a full proof.


\section{Dual representation}
\label{sect: dual representation}

In this section we establish the dual representation \eqref{eq: ES intro} of Expected Shortfall. As a preliminary step, we collect some elementary properties of Expected Shortfall that are used in the proof. They are direct consequences of elementary properties of quantile functions recorded in Lemma \ref{lem: quantiles} in the appendix. In particular, note that Expected Shortfall is well defined by Lemma \ref{lem: quantiles}(a) and is continuous from above by combining Lemma \ref{lem: quantiles}(d) with monotone convergence.

\begin{proposition}
\label{prop:ES}
For every $\alpha\in(0,1)$ the following statements hold:
\begin{enumerate}
    \item $\ES_\alpha(X)\in \R$ for every $X\in L^1$.
    \item $\ES_\alpha(X)\le\ES_\alpha(Y)$ for all $X,Y\in L^1$ such that $X\ge Y$ $\probp$-a.s.\ (monotonicity).
    \item $\ES_\alpha(X+c)=\ES_\alpha(X)-c$ for all $X\in L^1$ and $c\in\R$ (cash invariance).
    \item $\ES_\alpha(X_n)\uparrow\ES_\alpha(X)$ for all $(X_n)\subset L^1$ and $X\in L^1$ with $X_n\downarrow X$ $\probp$-a.s.
    \item For every $X\in L^1$, there is $k\in\N$ such that $\ES_\alpha(X)=\ES_\alpha(\min\{X,m\})$ for $m\in[k,\infty)$.
\end{enumerate}
\end{proposition}

\smallskip

In order to prove \eqref{eq: ES intro}, we first establish a link between the sign of Expected Shortfall and the sign of expectations taken under probabilities in the dual set $\cP_\alpha$ from \eqref{eq:cP alpha}.
In the language of risk measures, this is equivalent to establishing a dual representation of the acceptance set associated with Expected Shortfall. 

\begin{proposition}
\label{prop: key inequality}
Let $\alpha\in(0,1)$. For every $X \in L^1$ the following statements hold:
\begin{enumerate}
    \item If $\ES_\alpha(X)\le0$, then $\E_\probq\left(X\right)\ge0$ for every $\Q \in \cP_\alpha$. 
    \item If $\ES_\alpha(X) > 0$, then $\E_\probq\left(X\right)<0$ for some $\Q \in \cP_\alpha$.
\end{enumerate}
\end{proposition}

\begin{proof}
{\bf Step 1: Discrete random variables}. Let $X$ be a discrete random variable taking the values $x_1 < \cdots < x_N$ with probabilities $p_1, \ldots, p_N > 0$. For $\Q \in \cP_\alpha$ set $q_k :=  \Q(X = x_k)$ and note that $\sum_{k=1}^Nq_k=1$. Let $K=\min\{h\in\{1,\dots,N\} \,; \ \sum_{k =1}^h p_k \geq \alpha\}$. Then 
\begin{align}
\label{eq:pf:prop:key inequality:ES ineq}
\ES_\alpha(X) = \frac{1}{\alpha} \left(\sum_{k =1}^{K-1} -x_k p_k - x_K \bigg(\alpha - \sum_{k=1}^{K-1} p_k \bigg)\right)=  \sum_{k =1}^{K-1} (x_K - x_k) \frac{p_k}{\alpha} - x_K,
\end{align}
\begin{align}
\label{eq:pf:prop:key inequality:EQ eq}
\E_\probq(X) = \sum_{k=1}^Nx_kq_k = 
\sum_{k =1}^{K-1} -(x_K - x_k) q_k + x_K + \sum_{k =K+1}^N (x_k - x_K) q_k.
\end{align}

(a) Suppose that $\ES_\alpha(X)\le0$ and take $\Q \in \cP_\alpha$. As $x_K-x_k>0$ and $q_k \leq \frac{p_k}{\alpha}$ for every $k\in\{1,\dots,K-1\}$ and $x_k-x_K>0$ for every $k\in\{K+1,\dots,N\}$, \eqref{eq:pf:prop:key inequality:ES ineq} and \eqref{eq:pf:prop:key inequality:EQ eq} give
\[
\E_\probq(X) \ge \E_\probq(X)+\sum_{k =K+1}^N -(x_k - x_K) q_k = \sum_{k =1}^{K-1} -(x_K - x_k) q_k + x_K \ge -\ES_\alpha(X) \ge 0.
\]

(b) Suppose that $\ES_\alpha(X)>0$. We always find $\Q \in \cP_\alpha$ such that $q_k = \frac{1}{\alpha} p_k$ for $k \in \{1, \ldots, K-1\}$ and $q_k = 0$ for $k \in \{K+1, \ldots, N\}$. Then, \eqref{eq:pf:prop:key inequality:EQ eq} together with \eqref{eq:pf:prop:key inequality:ES ineq} give
\begin{equation*}
\E_\probq(X) = \sum_{k =1}^{K-1} -(x_K - x_k) q_k + x_K = \sum_{k =1}^{K-1} -(x_K - x_k) \frac{p_k}{\alpha} + x_K = - \ES_\alpha(X) < 0.
\end{equation*}

\smallskip

{\bf Step 2: Bounded random variables}. Let $X\in L^\infty$. Below we use the elementary fact that $X$ can be approximated uniformly from above by discrete random variables.

(a) Suppose that $\ES_\alpha(X)\le0$ but assume there is $\probq \in \cP_\alpha$ with $\E_\probq(X) < 0$. Then, we find a discrete random variable $X'$ satisfying $X-\frac{\E_\probq(X)}{2} \geq  X' \geq X$. This yields $\E_\probq(X') \leq \frac{\E_\probq(X)}{2} < 0$ and $\ES_\alpha(X') \leq \ES_\alpha(X) \leq 0$ by monotonicity of $\ES_\alpha$, which contradicts point (a) in Step 1.

(b) If $\ES_\alpha(X)>0$, then there exists a discrete random variable $X'$ with $X+\frac{\ES_\alpha(X)}{2} \geq  X' \geq X$. As a result, $\ES_\alpha(X') \geq \frac{\ES_\alpha(X)}{2} > 0$ by monotonicity and cash invariance of $\ES_\alpha$. It follows from point (b) in Step 1 that $\E_\probq(X) \leq \E_\probq(X')< 0$ for some $\probq \in \cP_\alpha$.

\medskip

{\bf Step 3: Integrable random variables}. Let $X\in L^1$ and for all $m,n\in\N$ define $X_{m,n}:=\max\{\min\{X,m\},-n\}\in L^\infty$. It follows from dominated convergence that $\E_\probp(|X-X_{m,n}|)\to0$. As for every $\probq\in\cP_\alpha$, we have $\frac{\diff \Q}{\diff P}\le\frac{1}{\alpha}$ $\probp$-a.s., this gives $\E_\probq(|X-X_{m,n}|)\to0$. Below we additionally use that, by Proposition~\ref{prop:ES}(e), there exists $k\in\N$ such that $\ES_\alpha(\min\{X,m\})=\ES_\alpha(X)$ for every $m\in\N$ with $m\ge k$.

(a) Assume that $\ES_\alpha(X)\le0$ and take any $\probq\in\cP_\alpha$. For all $m,n\in\N$ with $m\ge k$ we have $\ES_\alpha(X_{m,n})\le\ES_\alpha(\min\{X,m\})=\ES_\alpha(X)\le0$ by monotonicity of $\ES_\alpha$, whence $\E_\probq(X_{m,n})\ge0$ by (a) in Step 2. This yields $\E_\probq(X)\ge0$ as well.

(b) Suppose that $\ES_\alpha(X)>0$ and take any $\varepsilon\in(0,\ES_\alpha(X))$. Note that, for every $m\in\N$, we have $X_{m,n}+\varepsilon\downarrow \min\{X,m\}+\varepsilon$, which implies $\ES_\alpha(X_{m,n}+\varepsilon)\to\ES_\alpha(\min\{X,m\}+\varepsilon)$ by Proposition~\ref{prop:ES}(d). In particular, for every $m\in\N$ with $m\ge k$, we have $\ES_\alpha(X_{m,n}+\varepsilon)\to\ES_\alpha(X+\varepsilon)>0$ by cash invariance of $\ES_\alpha$. Hence, we can take $m,n\in\N$ large enough to obtain both $\ES_\alpha(X_{m,n}+\varepsilon)>0$ and $\E_\probp(|X-X_{m,n}|)<\alpha\varepsilon$. By (b) in Step 2, we find $\probq\in\cP_\alpha$ such that $\E_\probq(X_{m,n}+\varepsilon)<0$. It remains to observe that
\[
\E_\probq(X) \le \E_\probq(|X-X_{m,n}|)+\E_\probq(X_{m,n}) \le \frac{1}{\alpha}\E_\probp(|X-X_{m,n}|)-\varepsilon < 0.\qedhere
\]
\end{proof}

\smallskip

\noindent The preceding result delivers at once the desired representation of Expected Shortfall.

\begin{theorem}
\label{theo: dual representation}
Let $\alpha\in(0,1)$. For every $X \in L^1$ the following representation holds:
\begin{equation}
\label{eq:theo: dual representation}
\ES_{\alpha} (X) = \sup_{\Q \in \cP_\alpha} \E_\probq\left(-X\right).
\end{equation}
In particular, $\ES_\alpha$ is subadditive, i.e., for all $X,Y\in L^1$,
\[
\ES_\alpha(X+Y) \le \ES_\alpha(X)+\ES_\alpha(Y).
\]
\end{theorem}

\begin{proof}
Let $X\in L^1$. For $m \in \R$, by cash invariance of $\ES_\alpha$ and Proposition~\ref{prop: key inequality}(a) and (b),
\begin{align*}
\ES_\alpha(X) \leq  m&\;\; \Rightarrow \;\; \ \ES_\alpha(X+m) \leq 0 \;\; \Rightarrow \;\;
\inf_{\probq \in \cP_\alpha}\E_\probq(X+m)\ge0
\;\; \Rightarrow \;\;
\sup_{\probq \in \cP_\alpha}\E_\probq(-X )\leq m, \\
\ES_\alpha(X) > m&\;\; \Rightarrow \;\; \ \ES_\alpha(X+m) > 0 \;\; \Rightarrow \;\;
\inf_{\probq \in \cP_\alpha}\E_\probq(X+m) < 0
\;\; \Rightarrow \;\;
\sup_{\probq \in \cP_\alpha}\E_\probq(-X )>m.
\end{align*}
Combining the inequalities above yields \eqref{eq:theo: dual representation}. Finally, subadditivity follows directly from the the right-hand side of \eqref{eq:theo: dual representation} and subadditivity of the supremum.
\end{proof}

\smallskip

\begin{remark}
We have opted to divide the proof of Proposition \ref{prop: key inequality} into three steps to enhance versatility. If the interest is only on the dual representation of Expected Shortfall on a finite probability space or on a general probability space for bounded random variables, then one has to read only up to the end of  Step 1 or Step 2, respectively. The proof of Theorem \ref{theo: dual representation} is identical in these cases.
\end{remark}


\appendix

\section{Some properties of quantile functions}

Fix a probability space $(\Omega,\cF,\probp)$. For a random variable $X\in L^1$ we define the cumulative probability at $x\in\R$ by $F_X(x):=\probp(X\le x)$ and the upper quantile at level $\alpha\in(0,1)$ by
\[
q^+_\alpha(X) := \inf\{x\in\R \,; \ F_X(x)>\alpha\}.
\]
For a sequence $(X_n)\subset L^1$ we write $X_n\downarrow X$ whenever $X_n\ge X_{n+1}$ for every $n\in\N$ and $X_n\to X$ $\probp$-a.s.. Moreover, we set $X^+:=\max\{X,0\}$ and $X^-:=\max\{-X,0\}$. The next lemma collects some basic properties of quantile functions. We include a proof for completeness.

\begin{lemma}
\label{lem: quantiles}
For all $X\in L^1$ and $\alpha\in(0,1)$ the following statements hold:
\begin{enumerate}
    \item $q^+_\alpha(X)\in\R$ and $\int_0^1 |q^+_\beta(X)| \dd \beta < \infty$.
    \item $q^+_\alpha(X)\ge q^+_\alpha(Y)$ for all $X,Y\in L^1$ with $X\ge Y$ $\probp$-a.s.
    \item $q^+_\alpha(X+c)=q^+_\alpha(X)+c$ for all $X\in L^1$ and $c\in\R$.
    \item $q^+_\alpha(X_n)\downarrow q^+_\alpha(X)$ for every $(X_n)\subset L^1$ with $X_n\downarrow X$.
    \item There is $k\in\N$ such that $q^+_\beta(X)=q^+_\beta(\min\{X,m\})$ for all $\beta\in(0,\alpha)$ and $m\in\N$ with $m\ge k$.
\end{enumerate}
\end{lemma}
\begin{proof}
We only prove integrability in (a) and the assertions in (d) and (e) as the other statements are straightforward to verify by definition.

To establish integrability in (a), assume first that $X\ge0$ $\probp$-a.s., in which case $q^+_\beta(X)\ge0$ for every $\beta\in(0,1)$. For all $x\in\R$ and $\beta\in(0,1)$ we have
\[
q^+_\beta(X)>x \ \implies \ F_X(x)\le \beta \ \implies \ q^+_\beta(X)\ge x.
\]
By Fubini's theorem, we therefore obtain
\begin{align*}
\E(X) &= \int_0^\infty(1-F_X(x))\dd x = \int_0^\infty\int_0^1\ind_{[F_X(x),1]}(\beta)\dd \beta \dd x \\
&= \int_0^1\int_0^\infty\ind_{[0,q^+_\beta(X)]}(x) \dd x \dd \beta = \int_0^1q^+_\beta(X) \dd \beta,
\end{align*}
proving the desired assertion. For a generic $X\in L^1$ it suffices to observe that, for every $\beta\in(0,1)$, we have $q^+_\beta(X)=q^+_\beta(X^+)$ if $q^+_\beta(X)\ge0$ and $q^+_\beta(X)=q^+_\beta(-X^-)$ if $q^+_\beta(X)<0$, whence we derive that $|q^+_\beta(X)|=q^+_\beta(X^+)-q^+_\beta(-X^-)\le 2q^+_\beta(|X|)$ by (b).

To prove (d), take a sequence $(X_n)\subset L^1$ such that $X_n\downarrow X$ but assume that $q^+_\alpha(X)<\inf_{n\in\N}q^+_\alpha(X_n)$. Being increasing, $F_X$ has at most countably many discontinuity points. Hence, we find a continuity point $x\in\R$ for $F_X$ such that $q^+_\alpha(X)<x<\inf_{n\in\N}q^+_\alpha(X_n)$. As a result, $F_{X_n}(x)\le \alpha<F_X(x)$ for every $n\in\N$ but $F_{X_n}(x)\to F_X(x)$ by convergence in distribution, which is implied by almost-sure convergence. This is impossible and yields $q^+_\alpha(X_n)\downarrow q^+_\alpha(X)$ by (b).

Finally, to prove (e), take $\gamma\in(\alpha,1)$ and $k\in\N$ such that $k>q^+_\gamma(X)>q^+_\alpha(X)$. By (b), to conclude the proof, it suffices to show that, for all $\beta\in(0,\alpha)$ and $m\in\N$ with $m\ge k$, we have $q^+_\beta(X)\le q^+_\beta(\min\{X,m\})$. To this effect, take an arbitrary $x\in\R$ such that $F_{\min\{X,m\}}(x)>\beta$. As $q^+_\beta(\min\{X,m\})\le q^+_\beta(X)<q^+_\gamma(X)$ again by (b), we may assume without loss of generality that $x<q^+_\gamma(X)$. It is then easy to see that $F_X(x)=F_{\min\{X,m\}}(x)>\beta$, implying that $q^+_\beta(X)\le x$. This yields $q^+_\beta(X)\le q^+_\beta(\min\{X,m\})$ as desired.
\end{proof}


\end{document}